\documentclass[10pt,a4paper,twoside]{article}
	\title{A Weakly Intuitionistic Quantum Logic}
	\author{Ronnie Hermens}
	\usepackage{a4wide}
  \usepackage{soul} 
  \usepackage{url}
  \usepackage[small,bf,hang]{caption} 
  \usepackage{footmisc}
  \usepackage{footnote}
	\usepackage{marvosym} 
	\usepackage{amssymb}  
	\usepackage{extarrows}
	\usepackage{amsthm} 
	\usepackage{amsmath}  
	\usepackage{dsfont}
	\usepackage[english]{babel}    
	\usepackage[all,cmtip]{xy} 
	\DeclareMathOperator{\h}{\mathcal{H}}

	\DeclareMathOperator{\pee}{\mathds{P}}

	\DeclareMathAlphabet{\mathpzc}{OT1}{pzc}{m}{it} 	
	\newtheorem{stelling}{Theorem}[section]
	\newtheorem{lemma}[stelling]{Lemma}
	\newtheorem{gevolg}[stelling]{Corollary}
	
	\usepackage[T1]{fontenc} 
\begin{document}

\setlength{\unitlength}{\textwidth}

\maketitle

\begin{abstract}
In this paper we motivate and study the possibility of an intuitionistic quantum logic. An explicit investigation of the application of the theory of Bruns and Lakser on distributive hulls on traditional quantum logic (as suggested in \cite{Coecke02}) leads us to a small modification of this scheme. In this way we obtain a weak Heyting algebra (cf. \cite{Celani05}) for describing the language of quantum mechanics.
\end{abstract}

\section{An intuitionistic perspective on quantum logic}
Physical theories are concerned with statements about possible outcomes of experiments. A possible experiment may be termed an observable. In quantum mechanics, observables are identified with self-adjoint operators acting on a Hilbert space $\h$ whose domains are dense in $\h$. For a self-adjoint operator $A$, with spectrum $\sigma(A)$, we introduce the notation $A\in\Delta$ for some Borel set $\Delta\subset\sigma(A)$ for the statement that a measurement of $A$ will yield a result in $\Delta$ with probability one.\footnote{We will make no distinction in notation between observables and operators.} Quantum mechanics predicts that this is the case whenever the state of the system lies in the set $\mu_A(\Delta)\h$, where $\mu_A$ is the projective measure associated with the operator $A$. Consequently, statements of the form $A\in\Delta$ can be associated with closed linear subspaces of a Hilbert space. On the other hand, since there is a bijection between closed linear subspaces and projection operators, every closed linear subspace can be associated with a statement of the form $A\in\Delta$ (since projections are self-adjoint). 

This observation moved Birkhoff and von Neumann \cite{BirkhoffNeumann36} to introduce the quantum propositional lattice $L(\h)$ which consists of the set of closed linear subspaces of the Hilbert space $\h$ with partial order, meet and join defined in the following way:
\begin{itemize}
	\item $K_1\leq K_2$ iff $K_1\subset K_2$.
	\item $\bigwedge_{K\in \mathcal{K}}K:= \bigcap_{K\in \mathcal{K}}K$, $\mathcal{K}\subset L(\h)$.
	\item $\bigvee_{K\in \mathcal{K}}K:= \bigwedge\{K'\in L(\h)\:;\:K\leq K'\: \forall K\in\mathcal{K}\}$, $\mathcal{K}\subset L(\h)$.
\end{itemize}

A negation is also defined as
\begin{itemize}
	\item $\neg K:=\bigwedge\{K'\in L(\h)\:;\:K\vee K'=\h\}$.
\end{itemize}

The resulting lattice is almost a Boolean algebra, accept for the fact that the laws of distributivity
\begin{equation}
	K_1\vee(K_2\wedge K_3)=(K_1\vee K_2)\wedge(K_1\vee K_3)
\end{equation}
and 
\begin{equation}
	K_1\wedge(K_2\vee K_3)=(K_1\wedge K_2)\vee(K_1\wedge K_3)
\end{equation}
do not hold in general. As a consequence, it is hard to interpret the meet and join as the logical connectives ``and'' and ``or'' (cf. \cite{Dummett76}). Needless to say, quantum logic has struggled with interpretation problems ever since it was conceived. 

On the other hand, the fundamental problems of quantum mechanics have resulted in a consensus that quantum mechanics is incompatible with the logical structure of a classical phase space (cf. \cite{Isham95}). We, however, believe there may be a stronger discrepancy between quantum mechanics and classical logic. As an example we consider the following derivation of a Bell-type inequality.

\begin{lemma}\label{Belllemma}
Suppose $\pee$ is a probability function on a collection of sentences $S$ that satisfies the following rules for all $A,B\in S$:
\begin{enumerate}
\item If $A\to B$, then $\pee(A)\leq\pee(B)$.
\item $\pee(A\vee B)\leq\pee(A)+\pee(B)$.
\end{enumerate}
Then, if $S$ obeys classical logic, the following inequality holds for all $A_1,A_2,B_1$ and $B_2$ in $S$:
\begin{equation}\label{Bellineq}
	\pee(A_1\wedge B_1)\leq
	\pee(A_1\wedge B_2)+\pee(A_2\wedge B_1)+\pee(\neg A_2\wedge \neg B_2).
\end{equation}
\end{lemma}
\begin{proof}
The result follows by writing out in the following way:
\begin{equation}
\begin{split}
	\pee(A_1\wedge B_1)
	&=
	\pee(A_1\wedge B_1\wedge(B_2\vee\neg B_2))=\pee((A_1\wedge B_1\wedge B_2)\vee(A_1\wedge B_1\wedge\neg B_2))\\
	&\leq
	\pee(A_1\wedge B_1\wedge B_2)+\pee(A_1\wedge B_1\wedge\neg B_2)\leq\pee(A_1\wedge B_2)+\pee(B_1\wedge\neg B_2)\\
	&=
	\pee(A_1\wedge B_2)+\pee(B_1\wedge\neg B_2\wedge(A_2\vee\neg A_2))\\
	&=
	\pee(A_1\wedge B_2)+\pee((B_1\wedge\neg B_2\wedge A_2)\vee(B_1\wedge\neg B_2\wedge\neg A_2))\\
	&\leq
	\pee(A_1\wedge B_2)+\pee(B_1\wedge\neg B_2\wedge A_2)+\pee(B_1\wedge\neg B_2\wedge\neg A_2)\\
	&\leq
	\pee(A_1\wedge B_2)+\pee(A_2\wedge B_1)+\pee(\neg A_2\wedge \neg B_2).
\end{split}
\end{equation}	
\end{proof}
The inequality is however violated by quantum mechanics in the EPR-Bohm experiment \cite{Bohm57} of measurements on an entangled pair of spin-$\tfrac{1}{2}$ particles. In this setting, $A_1$ and $A_2$ are identified with two (mutually incompatible) measurements one experimenter can perform on one of the particles, and $B_1$ and $B_2$ with two (mutually incompatible) measurements an other experimenter can perform on the other particle. Each measurement has two possible outcomes, one regarded as a positive outcome (e.g. $A_1$), and the other as a negative outcome (e.g. $\neg A_1$).

The quantum logician may reject this inequality since its proof relies on an illegitimate use of the law of distributivity, but this is of course not the only solution to the paradox. In fact, quantum logic itself seems to hint towards an intuitionistic interpretation: the sentence $A_1\wedge B_1$ is both incompatible with $B_2$ and $\neg B_2$ and presents itself as an excluded middle. As Popper stated it:
\begin{quote}
	``It is of interest that the kind of change in classical logic which would fit what Birkhoff and von Neumann suggest [\ldots] would be the rejection of the law of excluded middle [\ldots], as proposed by Brouwer, but rejected by Birkhoff and von Neumann.'' \cite{Popper68}
\end{quote} 
This argument is perhaps a bit hand-waving, but it is of interest to note that the proof of Lemma \ref{Belllemma} also relies on an unlawful use of the law of excluded middle (at least from the point of view of the intuitionistic quantum logician).

The `discrepancy' that arises may be explained (somewhat sloppy) by observing that the negation as defined in quantum logic is somewhat intuitionisticaly in nature. That is, the negation of the statement $A\in\Delta$ is identified with the statement $A\in\Delta^c$, which is again a `positive' statement; it states that something will happen with probability 1. On the other hand, the definition of disjunction in quantum logic is typically non-intuitionistic and lies closer to the classical understanding of disjunction. Indeed, the truth of a disjunction $K_1\vee K_2$ does not imply the truth of either $K_1$ or $K_2$ in quantum logic. So, from an intuitionistic point of view, the disjunction in quantum logic is not a disjunction at all.

This is also roughly the viewpoint Coecke expresses in \cite{Coecke02}, and he argues that for an intuitionistic view on quantum mechanics
\begin{quote}
	``we formally need to introduce those additional propositions that express disjunctions of properties and that do not correspond to a property in the property lattice.'' 
\end{quote} 

These additional disjunctions are introduced by making use of Bruns and Lakser's theory of distributive hulls.
Concretely, this means that the quantum lattice $L(\h)$ is replaced by the lattice of distributive ideals of the quantum lattice:
\begin{equation} 
	\mathcal{DI}(L(\h)):=\{I\subset L(\h)\:;\: I\text{ is a distributive ideal}\},
\end{equation}
where by a distributive ideal we mean a non-empty subset $I$ such that
\begin{enumerate}
	\item if $K\in I$ and $K'\leq K$, then $K'\in I$.
	\item if $\mathcal{K}\subset I$ and for every $K'\in L(\h)$: $\left(\bigvee_{K\in\mathcal{K}}K\right)\wedge K'=\bigvee_{K\in\mathcal{K}}\left(K\wedge K'\right)$, then $\bigvee_{K\in\mathcal{K}}K\in I$.
\end{enumerate}

This new set is turned into a lattice by the following definitions:
\begin{itemize}
	\item $I_1\leq I_2$ iff $I_1\subset I_2$.
	\item $\bigwedge_{I\in\mathcal{I}}I:=\bigcap_{I\in\mathcal{I}}I$, $\mathcal{I}\subset\mathcal{DI}(L(\h))$.\footnote{One may show that this construction actually yields a distributive ideal.}
	\item $\bigvee_{I\in\mathcal{I}}I:=\bigwedge\{I'\in\mathcal{DI}(L(\h))\:;\:I\leq I' \forall I\in\mathcal{I}\}$, $\mathcal{I}\subset\mathcal{DI}(L(\h))$. 
\end{itemize}

With these definitions, $\mathcal{DI}(L(\h))$ is a complete distributive lattice. The propositions of the original lattice $L(\h)$ are identified with elements of $\mathcal{DI}(L(\h))$ by the injection 
\begin{equation}
	i:L(\h)\to\mathcal{DI}(L(\h)),\quad
	K\mapsto\downarrow K:=\{K'\in L(\h)\:;\: K'\leq K\}.
\end{equation}

As such, the construct of $\mathcal{DI}(L(\h))$ meets our desires; the new disjunction $\downarrow K_1\vee\downarrow K_2$ is not of the form $\downarrow K$ whenever $K_1\neq K_2$ and thus corresponds to a new element that does not correspond to any element in the original lattice. Because the new lattice is complete and the infinite laws of distributivity hold, it is also a complete Heyting algebra if one introduces the relative pseudo-complements:
\begin{itemize}
	\item $I_1\to I_2:=\bigvee\{I_3\in\mathcal{DI}(L(\h))\:;\: I_3\wedge I_1\leq I_2\}$.
\end{itemize}

\section{A classical perspective on quantum logic}

Complementary to the approach above, in stead of introducing a new disjunction that is more intuitionisticaly in nature than the one in quantum logic, on may want to define a new negation that is more classical in nature than the one in quantum logic. Recall that we took $A\in\Delta$ to stand for the statement that a measurement of $A$ will yield a result in $\Delta$ with probability one. In quantum logic, the negation of this proposition is the statement that a measurement of $A$ will yield a result in $\Delta^c$ with probability one. However, classically, the negation may be identified with the statement that a measurement of $A$ will not yield a result in $\Delta$ with probability one; i.e. one is not entirely certain that the measurement of $A$ will yield a result in $\Delta$. This statement is true for all the states in the set $(\h\backslash\mu_A(\Delta)\h)\cup\{0\}$. 

In this setting, it is easier to identify states with rays in the Hilbert space. Therefore we introduce the ray space
\begin{equation}
	R(\h):=\{[\psi]\:;\:\psi\in\h\backslash\{0\}\},\quad [\psi]:=\{\lambda\psi\:;\:\lambda\in\mathbb{C}\}.
\end{equation}
Propositions may then be identified with elements of the power set $\mathcal{P}(R(\h))$. Indeed, the proposition $A\in\Delta$ is now identified with the set $$\{[\psi]\in R(\h)\:;\:\psi\in\mu_A(\Delta)\h\}$$ and its negation, $\neg(A\in\Delta)$, with the complement of this set. The set $\mathcal{P}(R(\h))$ is turned into a lattice by introducing order, meet and join in the usual set-theoretic way:
\begin{itemize}
	\item $S_1\leq S_2$ iff $S_1\subset S_2$.
	\item $\bigwedge_{S\in\mathcal{S}}S:=\bigcap_{S\in\mathcal{S}}S$. 
	\item $\bigvee_{S\in\mathcal{S}}S:=\bigcup_{S\in\mathcal{S}}S$.  
\end{itemize}

Although this approach differs strongly from the intuitionistic approach, it is remarkable that both constructions are in fact identical:

\begin{stelling} The lattices $\mathcal{DI}(L(\h))$ and $\mathcal{P}(R(\h))$ are isomorphic (as complete bounded lattices). Consequently, the Heyting algebra $\mathcal{DI}(L(\h))$ is Boolean.\footnote{This second statement is in fact a consequence of the more general example below Lemma 1 in \cite{Coecke02}.}
\end{stelling}
\begin{proof}
We define the following function $f:\mathcal{P}(R(\h))\to\mathcal{P}(L(\h))$:
\begin{equation}
\begin{split}
	f(S)
	&:=\left\{K\in L(\h)\:;\:K\backslash\{0\}\subset\{\psi\in\h\:;\:[\psi]\in S\}\right\}\\
	&=\left\{K\in L(\h)\:;\:\bigcup_{\psi\in K\backslash\{0\}}\{[\psi]\}\subset S\right\}.
\end{split}
\end{equation}
Notice that it satisfies
\begin{equation}
	f(R(\h))=L(\h),\quad f(\varnothing)=\{0\}\quad\text{and}\quad f(\{[\psi]\})=[\psi]\quad\forall[\psi]\in R(\h).
\end{equation}
\indent

Now, for every $S\in\mathcal{P}(R(\h))$, $f(S)$ is in fact a distributive ideal. To show this, we have to show that $f(S)$ satisfies the properties 1 and 2. Suppose $K\in f(S)$ and $K'\leq K$. Then
\begin{equation}
	\bigcup_{\psi\in K'\backslash\{0\}}\{[\psi]\}\subset\bigcup_{\psi\in K\backslash\{0\}}\{[\psi]\}\subset S
\end{equation}
and thus $K'\in f(S)$.

To show property 2, we may assume $S\neq R(\h)$ (for $S=R(\h)$ 2 is trivially satisfied). Suppose $\mathcal{K}\subset f(S)$ such that for every $K'\in L(\h)$: $\left(\bigvee_{K\in\mathcal{K}}K\right)\wedge K'=\bigvee_{K\in\mathcal{K}}\left(K\wedge K'\right)$. We have to show that in that case $\bigvee_{K\in\mathcal{K}}K\in f(S)$. 

Suppose this isn't the case. Then there is a non-zero vector $\psi\in\bigvee_{K\in\mathcal{K}}K$ such that $[\psi]\notin S$. Furthermore $\psi\notin K$ for all $K\in\mathcal{K}$. But it then follows that
\begin{equation}
	[\psi]=\left(\bigvee_{K\in\mathcal{K}}K\right)\wedge [\psi]=\bigvee_{K\in\mathcal{K}}\left(K\wedge [\psi]\right)=\{0\}.
\end{equation}
This proves that $f:\mathcal{P}(R(\h))\to\mathcal{DI}(L(\h))$.

Next, consider the map
\begin{equation} 
	g:\mathcal{DI}(L(\h))\to\mathcal{P}(R(\h)),\quad g:I\mapsto\bigcup_{K\in I}\bigcup_{\psi\in K\backslash\{0\}}\{[\psi]\}.
\end{equation}
We will show that it is the inverse of $f$. First note that for every set $S\in\mathcal{P}(R(\h))$ one has
\begin{equation}
	g(f(S))=\bigcup_{K\in f(S)}\bigcup_{\psi\in K\backslash\{0\}}\{[\psi]\}\subset S.
\end{equation}
Now suppose $[\psi]\in S$, then $[\psi]\in f(S)$ and $[\psi]\in g(f(S))$. Thus $g(f(S))=S$ for all $S\in\mathcal{P}(R(\h))$.

So we have shown that $\mathcal{DI}(L(\h))$ and $\mathcal{P}(R(\h))$ are isomorphic as sets. However, since both $f$ and $g$ respect the partial order structure, it follows that $\mathcal{DI}(L(\h))$ and $\mathcal{P}(R(\h))$ are also isomorphic as complete lattices.
\end{proof}

\begin{gevolg}
There exists no probability function on $\mathcal{DI}(L(\h))$ that generalizes the Born rule, for by Lemma \ref{Belllemma}, any probability function would satisfy (\ref{Bellineq}).
\end{gevolg}

\section{A weakly intuitionistic perspective on quantum logic}


That the application of Bruns and Lakser's theory to the quantum lattice results in the construction of a Boolean algebra may be explained in the following way. The introduction of a new disjunction forces the introduction of a new negation. Indeed, the new negation in $\mathcal{DI}(L(\h))$ is defined as $\neg I:= I\to \downarrow 0$ and it is much weaker than the negation in quantum logic because one has
\begin{equation}
	\downarrow\neg K\leq\neg\downarrow K,\quad\forall K\in L(\h)
\end{equation}
with equality iff $K=0$ or $K=\h$. From the perspective of $\mathcal{P}(R(\h))$ it is clear to see that the negation in $\mathcal{DI}(L(\h))$ behaves classical rather than intuitionistic. 

It would seem more intuitionistic if one could generalize the negation of the quantum lattice to a negation in the lattice $\mathcal{DI}(L(\h))$. That is, by introducing a function $\sim:\mathcal{DI}(L(\h))\to\mathcal{DI}(L(\h))$ such that $\sim\downarrow K=\downarrow\neg K$ for all $K\in L(\h)$. In such a scheme, the negation of $A\in\Delta$ would coincide with $A\in\Delta^c$ like in quantum logic, but the disjunction of $A\in\Delta$ and $A\in\Delta^c$ would not be a triviality.\footnote{It is not unlikely that this scheme is also what Coecke envisaged in his paper.} This is in fact an idea explored in \cite[p 105--106]{Hermens10}.
Although in that text the emphasis is more on the set $\mathcal{P}(R(\h))$, the analysis is the same as for $\mathcal{DI}(L(\h))$ because one can use the embedding $r:L(\h)\to\mathcal{P}(R(\h))$ given by $r(K):=\{[\psi]\in R(\h)\:;\:\psi\in K\}$ for which the diagram
\begin{equation*}
	\xymatrix@!C{
		L(\h)\ar@{^{(}->}[rr]^i\ar@{^{(}->}[dd]^r & & \mathcal{DI}(L(\h))\ar@<1ex>[ddll]^{f^{-1}} \\
		& & \\
		\mathcal{P}(R(\h))\ar@<1ex>[uurr]^f & & }
\end{equation*}
commutes.

The generalization is straight forward. First note that $$r(\neg K)=\{[\psi]\:;\:\langle\psi,\phi\rangle=0\:\forall \phi\in K\}.$$ We therefore take
\begin{equation}
	\sim S:=\{[\psi]\in R(\h)\:;\:\langle\psi,\phi\rangle=0\:\forall \phi \text{ with }[\phi]\in S\}.
\end{equation}
Indeed, this results in $\sim r(K)=r(\neg K)$ for all $K\in L(\h)$. The `pseudo-negation' $\sim$ also behaves typically intuitionistic since we have
\begin{gather}
	S\vee\sim S=R(\h)\text{ iff }S=\varnothing\text{ or }S=R(\h),\\
	\sim S\vee\sim\sim S=R(\h)\text{ iff }S=\varnothing\text{ or }S=R(\h).
\end{gather}
However, one does have that
\begin{equation}
	\sim\sim(S\vee\sim S)=R(\h),\quad\forall S\in \mathcal{P}(R(\h)).
\end{equation}
One may also show that of the De Morgan laws only
\begin{equation}
	\sim S_1\wedge\sim S_2=\sim(S_1\vee S_2),\quad\forall S_1,S_2\in\mathcal{P}(R(\h))
\end{equation}
holds, and the other only holds in one direction:
\begin{equation}
	\sim S_1\vee\sim S_2\leq\sim(S_1\wedge S_2),\quad\forall S_1,S_2\in\mathcal{P}(R(\h)).
\end{equation}
The pseudo-negation also relates the `intuitionistic' disjunction of $\mathcal{P}(R(\h))$ to the `classical' disjunction of $L(\h)$ through the following equality:
\begin{equation}
	\sim\sim\left(\bigvee_{K\in\mathcal{K}}r(K)\right)=r\left(\bigvee_{K\in\mathcal{K}}K\right),\quad \text{for every}\quad \mathcal{K}\subset L(\h).
\end{equation}
So for any subset $S$ of $R(\h)$, its double pseudo-negation coincides with the closed linear subspace spanned by al the elements of $S$.

 Although the pseudo-negation appears to behave intuitionisticaly, there is no trivial way to incorporate the lattice $(\mathcal{P}(R(\h)),\vee,\wedge,\sim)$ in a Heyting algebra. This is because the relative pseudo-complement for the lattice $(\mathcal{P}(R(\h)),\vee,\wedge)$ is uniquely defined. There may however still be the possibility that a satisfactory implication relation $\to$ (that is not a relative pseudo-complement) may be defined on this lattice such that $S\to\bot=\sim S$ for all $S\in\mathcal{P}(R(\h))$. Indeed, we have the following result:
 
\begin{stelling}
There exists an implication relation such that $(\mathcal{P}(R(\h)),\vee,\wedge,\to)$ is a weakly Heyting algebra\footnote{The notion of weakly Heyting algebras was first introduced in \cite{Celani05}.}, i.e. a bounded distributive lattice in which for all $S_1,S_2,S_3\in\mathcal{P}(R(\h))$ one has
\begin{enumerate}
\item $S_1\to S_1=\top$,
\item $S_1\to(S_2\wedge S_3)=(S_1\to S_2)\wedge(S_1\to S_3)$,
\item $(S_1\vee S_2)\to S_3=(S_1\to S_3)\wedge (S_2\to S_3)$,
\item $(S_1\to S_2)\wedge (S_2\to S_3)\leq S_1\to S_3$,
\end{enumerate}
in such a way that for all $S\in\mathcal{P}(R(\h))$ 
\begin{equation}\label{(v)}
	S\to\bot=\sim S.
\end{equation}
\end{stelling} 
\begin{proof}
Let $\mathcal{P}_1$ denote the set of all atoms in $\mathcal{P}(R(\h))$ (note that there is a bijection between atoms and one-dimensional subspaces of $\h$). We now define
\begin{equation}\label{imp}
	S_1\to S_2:=\begin{cases}
	\top, &\text{if }S_1=\bot,\\
	\bigwedge_{\{s\in\mathcal{P}_1\:;\:s\leq S_1\}}\sim\sim(\sim s\vee(s\wedge S_2)), &\text{otherwise}.
	\end{cases}
\end{equation}
We will show that this implication relation satisfies the desired properties.

(i) Let $S\in\mathcal{P}(R(\h))$. If $S=\bot$, $S\to S=\top$ follows directly from the definition, so suppose $S\neq\bot$. In that case we have
\begin{equation}
	S\to S
	=
	\bigwedge_{\left\{\substack{s\in\mathcal{P}_1;\\s\leq S}\right\}}\sim\sim(\sim s\vee(s\wedge S))
	=
	\bigwedge_{\left\{\substack{s\in\mathcal{P}_1;\\s\leq S}\right\}}\sim\sim(\sim s\vee s)
	=
	\bigwedge_{\left\{\substack{s\in\mathcal{P}_1;\\s\leq S}\right\}}\top=\top
\end{equation}
Note that the same argument shows that if $S_1\leq S_2$ then $S_1\to S_2=\top$.
(ii) If $S_1=\bot$, the assertion is trivial. Suppose $S_1$ is an atom. We can distinct four scenarios: (1) $S_1\leq S_2\wedge S_3$, (2) $S_1\leq S_2$, $S_1\nleq S_3$, (3) $S_1\nleq S_2$, $S_1\leq S_3$ and (4) $S_1\nleq S_2$ and $S_1\nleq S_3$. In each of the cases it is easy to see that (ii) is satisfied. For all other $S_1$ we have
\begin{equation}
\begin{split}
	S_1\to(S_2\wedge S_3)
	&=
	\bigwedge_{\left\{\substack{s\in\mathcal{P}_1;\\s\leq S_1}\right\}}s\to(S_2\wedge S_3)
	=
	\bigwedge_{\left\{\substack{s\in\mathcal{P}_1;\\s\leq S_1}\right\}}(s\to S_2)\wedge (s\to S_3)\\
	&=
	(S_1\to S_2)\wedge (S_1\to S_3).
\end{split}
\end{equation}
(iii) If $S_1$ or $S_2$ equals $\bot$ the relation is again trivial, so suppose $S_1\neq\bot$ and $S_2\neq\bot$. We have
\begin{equation}
\begin{split}
	(S_1\vee S_2)\to S_3
	&=
	\bigwedge_{\left\{\substack{s\in\mathcal{P}_1;\\s\leq S_1\vee S_2}\right\}}s\to S_3
	=
	\left(\bigwedge_{\left\{\substack{s\in\mathcal{P}_1;\\s\leq S_1}\right\}}s\to S_3\right)\wedge
	\left(\bigwedge_{\left\{\substack{s\in\mathcal{P}_1;\\s'\leq S_2}\right\}}s'\to S_3\right)\\
	&=
	(S_1\to S_3)\wedge (S_2\to S_3).
\end{split}
\end{equation}
(iv) If $S_1=\bot$ the inequality follows immediately because then the right-hand side equals $\top$. The same also goes if $S_1\leq S_3$ so we suppose $S_1\nleq S_3$. If $S_1$ is an atom and $S_1\leq S_2$ we have $S_1\to S_2=\top$ (see proof of (i)) and
\begin{equation}
	(S_1\to S_2)\wedge (S_2\to S_3)=(S_2\to S_3)=\bigwedge_{\left\{\substack{s\in\mathcal{P}_1;\\s\leq S_2}\right\}}s\to S_3\leq S_1\to S_3.
\end{equation}
If $S_1\nleq S_2$, then 
\begin{equation}
	S_1\to S_2=\sim\sim(\sim S_1\vee (S_1\wedge S_2))=\sim\sim\sim S_1=\sim S_1.
\end{equation}
Similarly $S_1\to S_3=\sim S_1$ and thus
\begin{equation}
	(S_1\to S_2)\wedge (S_2\to S_3)=\sim S_1 \wedge S_2\to S_3\leq\sim S_1= S_1\to S_3.
\end{equation}
From this, the case where $S_1$ isn't an atom also follows:
\begin{equation}
\begin{split}
	(S_1\to S_2)\wedge (S_2\to S_3)
	&=
	\left(\bigwedge_{\left\{\substack{s\in\mathcal{P}_1;\\s\leq S_1}\right\}}s\to S_2\right)\wedge (S_2\to S_3)
	=
	\bigwedge_{\left\{\substack{s\in\mathcal{P}_1;\\s\leq S_1}\right\}}\left((s\to S_2)\wedge (S_2\to S_3)\right)\\
	&\leq
	\bigwedge_{\left\{\substack{s\in\mathcal{P}_1;\\s\leq S_1}\right\}}\left(s\to S_3\right)
	=
	S_1\to S_3.
\end{split}
\end{equation}
Finally, we have to show that (\ref{(v)}) holds:
\begin{equation}
	S\to\bot
	=
	\bigwedge_{\left\{\substack{s\in\mathcal{P}_1;\\s\leq S}\right\}}s\to\bot
	=
	\bigwedge_{\left\{\substack{s\in\mathcal{P}_1;\\s\leq S}\right\}}\sim s=\sim S.
\end{equation}
\end{proof}
 
However, it remains a difficult philosophical question what counts as a satisfactory implication relation for quantum logic, and it is not clear if (\ref{imp}) meets the requirements. It is also not clear if (\ref{imp}) is the unique implication relation that turns $(\mathcal{P}(R(\h)),\vee,\wedge)$ into a weakly Heyting algebra such that (\ref{(v)}) holds. And there is of course also still the question if a probability function can be defined on this lattice that generalizes the Born rule and explains the violation of the inequality (\ref{Bellineq}). Either way, we do believe that the weakly intuitionistic quantum logic defined here is philosophically at least a bit more satisfying than the logic of Birkhoff and von Neumann and perhaps even a step in the right direction for a comprehensible quantum logic.

\addcontentsline{toc}{section}{References}
\markboth{}{}
\bibliographystyle{rhalphanum}
\bibliography{referenties}

\end{document}